\documentclass[twocolumn,showpacs,amsmath,amssymb,10pt,aps,tikz,border=5mm]{revtex4}
  
\usepackage[utf8]{inputenc}
\usepackage{amsmath}
\usepackage{amsfonts}
\usepackage{graphicx}
\usepackage{yfonts}
\usepackage{wasysym}
\usepackage{color}
\usepackage[normalem]{ulem}
\usepackage{amsthm}
\usepackage{bm}
\usepackage{bbm}
\usepackage{mathtools}

\DeclareMathOperator{\imag}{i}

\DeclareMathOperator{\der}{d}

\def\<{\langle}
\def\>{\rangle}

\newtheorem{example}{Example}
\usepackage{placeins}

\newcommand{{\Cd}}{{\mathbb{C}^d}}

\newcommand{\Tr}{\mathrm{Tr}}

\def\oper{{\mathchoice{\rm 1\mskip-4mu l}{\rm 1\mskip-4mu l}
{\rm 1\mskip-4.5mu l}{\rm 1\mskip-5mu l}}}
\def\<{\langle}
\def\>{\rangle}
\newtheorem{Theorem}{Theorem}

\newtheorem{Proposition}{Proposition}

\newcommand{\beq}{\begin{equation}}
\newcommand{\eeq}{\end{equation}}
\newcommand{\bear}{\begin{eqnarray}}
\newcommand{\ear}{\end{eqnarray}}
\newcommand{\bdm}{\begin{displaymath}}
\newcommand{\edm}{\end{displaymath}}

\begin{document}

\title{\bf Memory kernel approach to  generalized Pauli channels: \\ Markovian, semi-Markov, and beyond}
\author{ Katarzyna Siudzi{\'n}ska and Dariusz Chru{\'s}ci{\'n}ski }
\affiliation{ Institute of Physics, Faculty of Physics, Astronomy and Informatics \\  Nicolaus Copernicus University,
Grudzi\c{a}dzka 5/7, 87--100 Toru\'n, Poland}

\begin{abstract}
In this paper, we analyze the evolution of the generalized Pauli channels governed by the memory kernel master equation. We provide
necessary and sufficient conditions for the memory kernel to give rise to the legitimate (completely positive and trace-preserving) quantum evolution.
In particular, we analyze a class of kernels generating the quantum semi-Markov evolution, which is a natural generalization of the Markovian semigroup. Interestingly, the convex combination of Markovian semigroups goes beyond the semi-Markov case. Our analysis is illustrated with several examples.
\end{abstract}

\maketitle

\section{Introduction}

In the theory of open quantum systems \cite{OS1,OS2,OS3}, the use of the Born-Markov approximation
leads to the celebrated Markovian master equation,
\begin{equation} \label{MME}
\dot{\rho}_t=\mathcal{L}[\rho_t],
\end{equation}
where $\mathcal{L}$ is the generator of the Markovian semigroup given by the well-known Gorini-Kossakowski-Sudarshan-Lindblad form \cite{GKS,L},
{\medmuskip=0.5mu
\thinmuskip=0.5mu
\thickmuskip=0.5mu
\begin{equation}\label{GKSL_gen}
\mathcal{L}[\rho]=-\imag[H_{\rm eff},\rho] + \frac 12 \sum_\alpha \gamma_\alpha \left( V_\alpha \rho V_\alpha^\dagger - \frac 12 \{ V_\alpha^\dagger V_\alpha,\rho\}_+ \right),
\end{equation}}
where $H_{\rm eff}$ is the effective Hamiltonian of the system, $V_\alpha$ denote the noise operators, and $\gamma_\alpha\geq 0$ are the decoherence rates. Equation (\ref{MME}) leads to the completely positive, trace-preserving (CPTP) dynamical map $\rho_0 \longrightarrow \rho_t = \Lambda_t[\rho_0]$ satisfying the composition law,
\begin{equation}\label{CL}
  \Lambda_{t} \Lambda_u = \Lambda_{t+u} ,
\end{equation}
for all $t,u \geq 0$. The Born-Markov approximation assumes weak interactions and a separation of time scales between the system and its environment. Such approximation is usually valid in quantum optical systems. However, it is often violated in solid state physics. There are two natural generalizations of the above scheme. The first one introduces the time-local generator $\mathcal{L}_t$ which is of the form (\ref{GKSL_gen}) but with time-dependent $V_\alpha(t)$ and $\gamma_\alpha(t)$. In the second approach, one takes into account non-local memory effects through the Nakajima-Zwanzig equation \cite{Nakajima,Zwanzig},
\begin{equation}\label{memory_kernel}
\dot{\rho}_t = \int_0^t K_{t-\tau} \rho_\tau \der\tau,
\end{equation}
with $K_t$ being the memory kernel. Indeed, contrary to (\ref{MME}),  the rate $\dot{\rho}_t$ at the time $t$ depends on the whole history $\rho_\tau$ -- starting from the initial time $\tau = 0$, up to the current time $\tau=t$.  The Markovian semigroup (\ref{GKSL_gen}) is recovered for $K_t=2\delta(t)\mathcal{L}$.




The central problem with the memory kernel master equation (\ref{memory_kernel}) is to provide the necessary and sufficient conditions for the memory kernel super-operator $K_t$ which guarantee that the solution in the form of the dynamical map $\Lambda_t$ is CPTP. Such problem was originally posed by Barnett and Stenholm \cite{Barnett} for the memory kernel
$$ K_t= k(t)\mathcal{L} $$
with the memory function $k(t)$ and the legitimate Markovian generator $\mathcal{L}$. Unfortunately, in general, such memory kernels may lead to unphysical results. This issue was further analyzed in \cite{Maniscalco, Budini}.  Shabani and Lidar \cite{Shabani} proposed the so-called {\it post-Markovian master equation} with 
$$ K_t=k(t)\mathcal{L}e^{\mathcal{L}t} . $$
Again, this approach works for certain classes of Markovian generators $\mathcal{L}$ and memory functions $k(t)$. 
The authors succeeded in finding the necessary and sufficient conditions for this memory kernel to be legitimate. 
There is also the class of the qubit evolution \cite{Petruccione} for which this kernel always produces physical results.
Much attention was paid to finding the admissible memory kernels. It turned out that they can arise from the collisional model \cite{Vacchini2}. Another class was found for the random unitary qubit evolution \cite{chlopaki}. The quantum analogue of the semi-Markov evolution was analyzed in \cite{Piilo,Breuer,Breuer_Vacchini}. Interestingly, the proper definition of the quantum semi-Markov evolution was given only in \cite{semi-2}, using the notion of legitimate pairs of quantum maps \cite{early_semi-Markov}. For recent papers discussing memory kernel approach see also \cite{Ciccarello,Vacchini,Vacchini-2016,DC-2016,semi-1,Palermo}.

In this paper, we analyze the evolution of the generalized Pauli channels under the memory kernel master equation (\ref{memory_kernel}). We provide the necessary and sufficient conditions for the admissible memory kernel -- that is, the kernel giving rise to the CPTP dynamical map $\Lambda_t$. A special class of memory kernels corresponds to the so-called semi-Markov quantum evolution, which is the quantum analogue of the classical semi-Markov process. We provide several examples of the semi-Markov evolution of the generalized Pauli channels. Interestingly, the convex combination of Markovian semigroups (which is also the generalized Pauli channel) is not semi-Markov.

\section{Generalized Pauli channels}

The definition of the generalized Pauli channel involves the notion of mutually unbiased bases (MUBs). Two orthonormal bases $|\psi_k\>$, $|\phi_l\>\ \in\Cd$ are said to be mutually unbiased if and only if
\begin{equation}
|\<\psi_k|\phi_l\>|^2 = \frac 1 d .
\end{equation}
For $d=p^r$, where $p$ is a prime number, the number of MUBs in $\mathbb{C}^d$ is maximal and equal to $d+1$ \cite{Wootters,MAX}.

Take the $d$-dimensional Hilbert space for which one has $d+1$ MUBs,
$\{ |\psi^{(\alpha)}_0\>,\ldots,|\psi^{(\alpha)}_{d-1}\> \}$. The corresponding rank-1 projectors are given by $P^{(\alpha)}_l = |\psi^{(\alpha)}_l\>\< \psi^{(\alpha)}_l|$. Now, let us define  $d+1$ unitary operators
\begin{equation}\label{U}
U_{\alpha} = \sum_{l=0}^{d-1} \omega^{l} P_l^{(\alpha)} \ ,
\end{equation}
where $\omega = e^{2\pi i/d}$, and the family of completely positive maps
\begin{equation}
\mathbb{U}_\alpha[\rho] = \sum_{k=1}^{d-1}  U_{\alpha}^k \rho  U_{\alpha}^{k \dagger} .
\end{equation}
The evolution under the generalized Pauli channel is given by the following dynamical map \cite{Ruskai,nasze},
\begin{equation}\label{GPC}
\Lambda_t = p_0(t) \oper + \frac{1}{d-1}\sum_{\alpha=1}^{d+1} p_\alpha(t) \mathbb{U}_\alpha ,
\end{equation}
where $(p_0(t),p_1(t),\ldots,p_{d+1}(t))$ denotes the probability vector such that $p_0(0)=1$ and $p_\alpha(0)=0$ for $\alpha=1,\ldots,d+1$. By the identity map $\oper$, we understand $\oper[X]=X$ for any operator $X$. It is clear that  this definition reproduces the Pauli channel for $d=2$,
\begin{equation}\label{PC}
\Lambda_t = p_0(t) \oper + \sum_{\alpha=1}^{3} p_\alpha(t) \mathbb{U}_\alpha ,
\end{equation}
with $\mathbb{U}_\alpha[\rho] = \sigma_\alpha \rho \sigma_\alpha$, and $\sigma_\alpha$ being the Pauli matrices.

One easily solves the eigenvalue problem for $\Lambda_t$,
\begin{equation}\label{GPC_eigenvalue_eq}
\Lambda_t[ U_\alpha^k] = \lambda_\alpha(t) U_\alpha^k  \ , \ \ k=1,\ldots,d-1 ,
\end{equation}
with the eigenvalues
\begin{equation}\label{GPC_eigenvalues}
\lambda_\alpha(t) = p_0(t) + \frac{d}{d-1}p_\alpha(t) - \frac{1}{d-1} \sum_{\beta=1}^{d+1} p_\beta(t) ,
\end{equation}
and $\lambda_0(t)=1$. All the eigenvalues are real, whereas $\lambda_\alpha(t)$ ($\alpha=1,\dots,d+1$) are $(d-1)$-fold degenerated. The inverse relation reads
\begin{equation}\label{c1}
p_0(t)=\frac{1}{d^2}\left[1+(d-1)\sum_{\alpha=1}^{d+1}\lambda_\alpha(t)\right],
\end{equation}
\begin{equation}\label{c2}
p_\alpha(t) = \frac{d-1}{d^2} \left( 1 + d\lambda_\alpha(t) - \sum_{\beta=1}^{d+1} \lambda_\beta \right).
\end{equation}
Equations (\ref{c1}-\ref{c2}) make it clear that $\Lambda_t$ is a completely positive map if and only if the direct generalization of the Fujiwara-Algoet conditions \cite{Fujiwara, Ruskai, Zyczkowski},
\begin{equation}\label{Fuji-d}
-\frac{1}{d-1} \leq  \sum_{\beta =1}^{d+1} \lambda_\beta(t) \leq 1 + d \min_\beta \lambda_\beta(t) ,
\end{equation}
is satisfied for all $t \geq 0$.

The map $\Lambda_t$ satisfies the time-local master equation
\begin{equation}
\dot{\Lambda}_t = \mathcal{L}_t \Lambda_t
\end{equation}
with the corresponding time-local generator
\begin{equation}\label{GEN}
\mathcal{L}_t = \sum_{\alpha=1}^{d+1} \gamma_\alpha(t) \mathcal{L}_\alpha
\end{equation}
and
\begin{equation}
\mathcal{L}_\alpha = \frac 1d \left[ \mathbb{U}_\alpha -(d-1)\oper \right] = \Phi_\alpha-\oper ,
\end{equation}
where $\Phi_\alpha$ define the family of depolarizing channels,
\begin{equation}
  \Phi_\alpha[\rho] =  \sum_{l=0}^{d-1}  P_l^{(\alpha)} \rho  P_l^{(\alpha)}.
\end{equation}
The eigenvalue equation for $\mathcal{L}_t$ reads
\begin{equation}
\mathcal{L}_t[U^k_\alpha] = \mu_\alpha(t) U^k_\alpha ,
\end{equation}
with $\mu_\alpha(t) = \gamma_\alpha(t) - \gamma(t)$ and $\gamma(t) = \sum_{\alpha=1}^{d+1}\gamma_\alpha(t)$. Therefore, the time-dependent eigenvalues $\lambda_\alpha(t)$ of the dynamical map $\Lambda_t$ are given by
\begin{equation}\label{eigenvalue2}
\lambda_\alpha(t) = \exp[\Gamma_\alpha(t) - \Gamma(t)] ,
\end{equation}
where $\Gamma_\alpha(t)=\int_0^t\gamma_\alpha(\tau)\der\tau$ and $\Gamma(t)= \sum_{\alpha=1}^{d+1}\Gamma_\alpha(t)$.

\section{Memory kernel approach}

In this paper, we analyze the evolution of the generalized Pauli channel $\Lambda_t$ which is provided by the memory kernel equation
\begin{equation}\label{memory_kernel_equation}
\dot{\Lambda}_t=\int_0^tK_{t-\tau}\Lambda_\tau\der\tau
\end{equation}
with the following memory kernel,
\begin{equation}\label{K}
K_t = \sum_{\alpha=1}^{d+1}k_\alpha(t)\left[
\Phi_\alpha-\oper\right].
\end{equation}
Note that the eigenvalue equations of such a memory kernel are given by
\begin{equation}\label{kernel_eigenvalue_eq}
K_t[U_\alpha^k]=\kappa_\alpha(t) U_\alpha^k, \qquad K_t[\mathbb{I}]=0,
\end{equation}
where
\begin{equation}
\kappa_\alpha(t)=k_\alpha(t)- k(t) ,
\end{equation}
with $k(t) = \sum_{\beta=1}^{d+1}k_\beta(t)$. Taking (\ref{kernel_eigenvalue_eq}) and (\ref{GPC_eigenvalue_eq}) into account, we can rewrite the relationship between the memory kernel $K_t$ and the corresponding generalized Pauli channel $\Lambda_t$ (\ref{memory_kernel_equation}) in terms of the corresponding eigenvalues,
\begin{equation}
\dot{\lambda}_\alpha(t)=\int_0^t\kappa_\alpha(t-\tau) \lambda_\alpha(\tau)\der\tau ,
\end{equation}
with $\lambda_\alpha(0)=1$. In the Laplace transform (LT) domain, one finds the following relation,
\begin{equation}\label{lambda_s}
\widetilde{\lambda}_\alpha(s)=\frac{1}{s-\widetilde{\kappa}_\alpha(s)},
\end{equation}
where $\widetilde{f}(s)=\int_0^\infty f(t)e^{-st} \der t$ stands for the Laplace transform of $f(t)$. Let us parameterize the eigenvalues $\lambda_\alpha(t)$ as follows,
\begin{equation}\label{}
  \lambda_\alpha(t) = 1 - \int_0^t \ell_\alpha(\tau) \der \tau .
\end{equation}
One arrives at the following theorem.

\begin{Theorem} \label{TH1} The memory kernel $K_t$ defined in (\ref{K}) gives rise to a legitimate dynamical map $\Lambda_t$ if and only if the corresponding eigenvalues $\kappa_\alpha(t)$ are, in the LT domain, given by
\begin{equation}\label{th_1}
  \widetilde{\kappa}_\alpha(s) = - \frac{s \widetilde{\ell}_\alpha(s)}{1-   \widetilde{\ell}_\alpha(s)} ,
\end{equation}
and the functions $\ell_\alpha(t)$ satisfy
\begin{eqnarray}\label{CON}
   \int_0^t \ell_\alpha(\tau) \der \tau &\geq & 0, \nonumber \\
  \sum_{\alpha=1}^{d+1}   \int_0^t \ell_\alpha(\tau) \der \tau &\leq & \frac{d^2}{d-1}, \\
   \sum_{\alpha=1}^{d+1}   \int_0^t \ell_\alpha(\tau) \der \tau & \geq &  d\, \int_0^t \ell_\beta(\tau) \der \tau \nonumber
\end{eqnarray}
for $\beta=1,\ldots,d+1$.
\end{Theorem}
The proof is evident since the above conditions reproduce $\lambda_\alpha(t) \leq 1$, $p_0(t) \geq 0$, and $p_\beta(t) \geq 0$ for $\beta=1,\ldots,d+1$, respectively. The main problem is to find a reasonable class of functions $\ell_\alpha(t)$ satisfying conditions (\ref{CON}).

\begin{Proposition} Consider $\ell_\alpha(t) = \eta e^{- \xi_\alpha t}$ with $\eta, \xi_\alpha > 0$. If
\begin{eqnarray}
 \label{1a}
  \eta \sum_{\alpha=1}^{d+1} \frac{1}{\xi_\alpha}  &\leq& \frac{d^2}{d-1} , \\ \label{2a}
   \sum_{\alpha=1}^{d+1} \frac{1}{\xi_\alpha} &\geq &  \frac{d}{\xi_\beta} ,
\end{eqnarray}
then $\ell_\alpha(t)$'s satisfy (\ref{CON}).
\end{Proposition}

\begin{proof}
Let us start with showing that the first inequality in the Fujiwara-Algoet conditions (\ref{Fuji-d}),
\begin{eqnarray}\label{fuji}
 -\frac{1}{d-1} \leq  \sum_{\beta =1}^{d+1} \lambda_\beta,
\end{eqnarray}
is equivalent to (\ref{1a}). For our choice of $\ell_\alpha(t)$'s, the eigenvalues of $\Lambda_t$ are equal to
\begin{equation}\label{lambda}
\lambda_\alpha(t)=1-\frac{\eta}{\xi_\alpha}\left(1-e^{-\xi_\alpha t}\right).
\end{equation}
After inserting (\ref{lambda}) into (\ref{fuji}), one has
\begin{eqnarray}
&& -\frac{1}{d-1} \leq  d+1-\sum_{\alpha =1}^{d+1}\frac{\eta}{\xi_\alpha}\left(1-e^{-\xi_\alpha t}\right).
\end{eqnarray}
While this inequality holds for all $t\geq 0$, it is enough to check that it is true for $t\to\infty$. Therefore, we arrive at
\begin{equation}
-\frac{1}{d-1} \leq d+1-\sum_{\beta =1}^{d+1}\frac{\eta}{\xi_\beta},
\end{equation}
which is, indeed, equivalent to (\ref{1a}).

Now, we start from (\ref{2a}). Denote the minimal value of $\xi_\alpha$ by $\xi_{min}=\min_\alpha\xi_\alpha$. If (\ref{2a}) holds for every $\beta=1,\dots,d+1$, then it is also true for $\xi_{min}$. Multiplying both sides of the inequality by the same (positive) coefficient, we get
\begin{equation}\label{2aa}
\sum_{\alpha=1}^{d+1}\frac{1}{\xi_\alpha}\left(1-e^{-\xi_{min} t}\right)\geq \frac{d}{\xi_{min}}\left(1-e^{-\xi_{min} t}\right).
\end{equation}
Observe that
\begin{equation}
\sum_{\alpha=1}^{d+1}\frac{1}{\xi_\alpha}\left(1-e^{-\xi_\alpha t}\right)\geq \sum_{\alpha=1}^{d+1}\frac{1}{\xi_\alpha}\left(1-e^{-\xi_{min} t}\right),
\end{equation}
and therefore (\ref{2aa}) reduces to
\begin{equation}\label{ost}
\sum_{\alpha=1}^{d+1}\frac{1}{\xi_\alpha}\left(1-e^{-\xi_\alpha t}\right)\geq \frac{d}{\xi_{min}}\left(1-e^{-\xi_{min} t}\right).
\end{equation}
Lastly, note that
\begin{equation}
h(\xi,t)=\frac{d}{\xi}\left(1-e^{-\xi t}\right)
\end{equation}
is a function of $\xi$ which decreases monotonically with the increasing value of $\xi$ for each fixed $t\geq 0$. Hence, $h(\xi_{min},t)=\max_\alpha h(\xi_\alpha,t)$, which means that (\ref{ost}) is equivalent to the second inequality in the Fujiwara-Algoet conditions (\ref{Fuji-d}).
\end{proof}

For $\ell_\alpha(t) = \eta e^{- \xi_\alpha t}$, one finds
\begin{equation}\label{}
  \kappa_\alpha(t) = -\eta\delta(t)+\eta(\xi_\alpha-\eta)e^{-(\xi_\alpha-\eta)t},
\end{equation}
and finally
\begin{equation}\label{}
\begin{split}
  k_\alpha(t) =& \frac 1d \eta\delta(t)
+\eta(\xi_\alpha-\eta)e^{-(\xi_\alpha-\eta)t}
\\&-\frac{1}{d}\sum_{\beta=1}^{d+1}\eta(\xi_\beta-\eta)e^{-(\xi_\beta-\eta)t},
\end{split}
\end{equation}
which shows that a linear combination of simple exponential memory functions has to satisfy strong constraints (\ref{1a}--\ref{2a}). Observe that $\kappa_\alpha(t\to\infty)\to\infty$ for $\xi_\alpha-\eta<0$. Conditions (\ref{1a}-\ref{2a}) imply
\begin{equation}
\xi_\alpha-\eta\geq-\frac{\eta}{d},
\end{equation}
which means that we need an additional restriction for the choice of $\eta$ and $\xi_\alpha$'s to obtain a physical memory kernel.

%

A special class of memory kernels is given by
\begin{equation}\label{special}
  \ell_\alpha(t) = \frac{1}{a_\alpha} \ell(t) .
\end{equation}
In this case, Theorem \ref{TH1} implies the following.

\begin{Proposition} If the function $\ell(t)$ and the collection of numbers $\{a_1,\ldots,a_{d+1}\}$ satisfy
\begin{equation}\label{29}
\sum_{\alpha=1}^{d+1}\frac{1}{a_\alpha}\int_0^t \ell(\tau)\der\tau\leq \frac{d^2}{d-1},
\end{equation}
together with
\begin{equation}\label{30}
\sum_{\beta=1}^{d+1}\frac{1}{a_\beta}\geq\frac{d}{a_\alpha},
\end{equation}
then $K_t$ given by the following eigenvalues (in the LT domain),
\begin{equation}\label{kappa_tilde}
\widetilde{\kappa}_\alpha(s)=-\frac{s\widetilde{\ell}(s)/a_\alpha}{1-\widetilde{\ell}(s)/a_\alpha},
\end{equation}
defines the legitimate memory kernel for the evolution described by the generalized Pauli channel.
\end{Proposition}

\begin{Proposition}
Let $\ell(t)$ be given by the following convolution,
\begin{equation}
\ell(t)=e^{-z_1t}*\cdots*e^{-z_nt},
\end{equation}
where $z_k > 0$ and $z_i \neq z_j$ for $i \neq j$. If $a_\alpha$ satisfy (\ref{30}) and
\begin{equation}\label{56}
\prod_{k=1}^nz_k\geq \frac{d-1}{d^2} \sum_{\alpha=1}^{d+1}\frac{1}{a_\alpha},
\end{equation}
then $\widetilde{\kappa}_\alpha(s)$'s given in (\ref{kappa_tilde}) define a legitimate memory kernel.
\end{Proposition}

\begin{proof}
It is enough to verify that $p_0(t\to\infty)\geq 0$, as the smallest value of $p_0(t)$ corresponds to the asymptotic case where $t\to\infty$. Using the following result,
\begin{equation}
\int_0^\infty e^{-z_1t} \ast \ldots \ast e^{-z_nt}\der t=\frac{1}{z_1 \ldots z_n},
\end{equation}
together with equation (\ref{c1}), let us write out the explicit form of $p_0(t\to\infty)$ for the chosen $\ell(t)$,
\begin{equation}\label{q0}
p_0(t\to\infty)=d^2-(d-1)\sum_{\alpha=1}^{d+1}\frac{1}{a_\alpha}
\frac{1}{\prod_{k=1}^nz_k}.
\end{equation}
This is always non-negative if (\ref{56}) is satisfied.
\end{proof}

\section{Semi-Markov evolution}

The quantum semi-Markov evolution is the quantum analogue of the classical concept of the stochastic semi-Markov process.
Such process is defined  in terms of the semi-Markov matrix $q_{ij}(t)\geq 0$ ($t\geq 0$), which determines the probability $\int_0^tq_{ij}(\tau)\der\tau$ of jump $j\to i$ at $\tau\in[0,t]$ if the system is in the state $j$ at $\tau=0$. Using this matrix, one defines the waiting time distribution and the survival probability by
\begin{equation}
f_j(t)=\sum_{i=1}^dq_{ij}(t),\qquad g_j(t)=1-\int_0^tf_j(\tau)\der\tau,
\end{equation}
respectively. The stochastic evolution of the probability vector $\mathbf{p}$,
\begin{equation}
\mathbf{p}(t)=T(t)\mathbf{p},\qquad T(0)=\oper,
\end{equation}
is provided by the stochastic map constructed as follows,
\begin{equation}
T(t)=n(t)+(n*q)(t)+(n*q*q)(t)+\dots,
\end{equation}
where $n_{ij}(t) = g_j(t) \delta_{ij}$. It satisfies classical memory kernel master equation,
\begin{equation}\label{}
  \frac{d}{dt} {T}(t) = \int_0^t K(t-\tau) T(\tau) d\tau ,
\end{equation}
with
\begin{equation}\label{}
  \widetilde{K}(s) = s \mathbb{I} - [\mathbb{I} - \widetilde{q}(s)] \widetilde{n}^{-1}(s) .
\end{equation}

The quantum semi-Markov evolution \cite{semi-1,semi-2} is defined in terms of the so-called quantum semi-Markov map, i.e. the completely positive map $Q_t$ for which $\int_0^t Q_\tau^\dagger[\mathbb{I}]\der\tau\leq\mathbb{I}$. By $Q_t^\dagger$, we understand the map dual to $Q_t$ in the sense that $\Tr(XQ_t[Y])=\Tr(Q_t^\dagger[X]Y)$.
For the given semi-Markov map, one defines the waiting time operator $F_t=Q_t^\dagger[\mathbb{I}]$ and the survival operator
\begin{equation}
G_t=\mathbb{I}-\int_0^tF_\tau\der\tau,
\end{equation}
where $G_t\geq 0$, $G_0=\mathbb{I}$. In the quantum semi-Markov evolution, $N_t$ is given by
\begin{equation}
N_t[\rho]=\sqrt{G_t}\rho\sqrt{G_t},
\end{equation}
and therefore it is fully determined by the choice of $Q_t$. The dynamical map $\Lambda_t$ is represented by the series of convolutions,
\begin{equation}\label{Dyson}
  \Lambda_t = N_t + N_t \ast Q_t + N_t \ast Q_t \ast Q_t + \ldots .
\end{equation}
This series is convergent if $||\widetilde{Q}_s||_1<1$, where $||X||_1$ denotes the trace norm of $X$. Such representation of the dynamical map allows us to construct the corresponding memory kernel via
\begin{equation}\label{Kt}
K_t=B_t-Z_t.
\end{equation}
The maps $B_t$ and $Z_t$ are defined, in the LT domain, by the following relations,
\begin{equation}\label{}
\widetilde{N}_s=[s\oper+\widetilde{Z}_s]^{-1},\qquad \widetilde{Q}_s=\widetilde{B}_s\widetilde{N}_s,
\end{equation}
and give rise to the following formula for the memory kernel,
\begin{equation}\label{}
  \widetilde{K}_s = s \oper - [\oper - \widetilde{Q}_s]\widetilde{N}_s^{-1} .
\end{equation}
Now, for the generalized Pauli channels, we take
\begin{equation}\label{Q_t}
Q_t = \frac{1}{d-1}\sum_{\alpha=1}^{d+1} f_\alpha(t)
 \mathbb{U}_\alpha,
\end{equation}
with $f_\alpha(t)\geq 0$ and $\int_0^\infty f(t) \der t \leq 1$, where
\begin{equation}\label{}
  f(t) = \sum_{\alpha=1}^{d+1}f_\alpha(t).
\end{equation}
The quantum waiting time and the quantum survival time operators have simple forms,
\begin{equation}
F_t= f(t) \mathbb{I},\qquad
G_t=g(t)\mathbb{I},
\end{equation}
with
\begin{equation}
g(t)=1- \int_0^t f(\tau)\der\tau.
\end{equation}
After some straightforward calculations, we obtain the following semi-Markov memory kernel,
\begin{equation}\label{semi-Markov_kernel}
{K}_t = \sum_{\alpha=1}^{d+1} {k}_\alpha(t)
\left[\Phi_\alpha-\oper\right] ,
\end{equation}
where
\begin{equation}
\widetilde{k}_\alpha(s) = \frac{d}{d-1}\frac{\widetilde{f}_\alpha(s)}{\widetilde{g}(s)} .
\end{equation}
Finally, the generalized Pauli channel generated by (\ref{semi-Markov_kernel}) is determined by
\begin{equation}\label{eigenv_SM}
\widetilde{\lambda}_\alpha(s)=-\frac{d-1}{s}\frac{\widetilde{f}(s)-1}
{\widetilde{f}(s) - d \widetilde{f}_\alpha(s) + d  -1}.
\end{equation}
It implies the following relations between $\widetilde{f}_\alpha(s)$ and $\widetilde{\ell}_\alpha(s)$:
\begin{equation}
\widetilde{\ell}_\alpha(s)=\frac{d\left(\widetilde{f}(s)-\widetilde{f}_\alpha(s)\right)}
{\widetilde{f}(s)- d\widetilde{f}_\alpha(s) + d-1},
\end{equation}
\begin{equation}
\widetilde{f}_\alpha(s)=\frac{\sum_{\beta=1}^{d+1}\frac{1}{1-\widetilde{\ell}_\beta(s)}
-\frac{d}{1-\widetilde{\ell}_\alpha(s)}-1}
{\sum_{\beta=1}^{d+1}\frac{1}{1-\widetilde{\ell}_\beta(s)}+\frac{1}{d-1}}.
\end{equation}

In the isotropic case -- that is, when $f_\alpha(t) = \chi(t)$ and
$$   \int_0^\infty \chi(t) \der t \leq \frac{1}{d+1}, $$
one finds
\begin{equation}\label{}
  \widetilde{\ell}_\alpha(s) = \widetilde{\nu}(s) = \frac{d^2 \widetilde{\chi}(s)}{\widetilde{\chi}(s) + d-1} .
\end{equation}

It turns out \cite{semi-1} that the representation (\ref{Dyson}) of $\Lambda_t$ allows one to consider the following inhomogeneous memory kernel master equation \cite{Ciccarello,Vacchini},
\begin{equation}\label{IN}
  \dot{\Lambda_t} = \int_0^t \mathbb{K}_{t-\tau} \Lambda_\tau d \tau + \dot{N}_t,
\end{equation}
where the new kernel $\mathbb{K}_t$ is defined by
\begin{equation}\label{}
  \widetilde{\mathbb{K}}_s = s \widetilde{N}_s \widetilde{Q}_s \widetilde{N}^{-1}_s  .
\end{equation}
In particular, if $\widetilde{N}_s$ and $\widetilde{Q}_s$ commute, then $\widetilde{\mathbb{K}}_s = s  \widetilde{Q}_s $ -- or, equivalently, in the time domain, $\mathbb{K}_t = \dot{Q}_t + Q_0\delta(t)$. In our case, it gives
\begin{equation}\label{}
  \mathbb{K}_t = \sum_{\alpha=1}^{d+1} h_\alpha(t) \mathbb{U}_\alpha ,
\end{equation}
with $h_\alpha(t)=\dot{f}_\alpha(t)+f_\alpha(0)\delta(t)$,
and hence eq. (\ref{IN}) provides the following inhomogeneous equation for the density operator $\rho_t$ with the initial state $\rho_0$,
\begin{equation}\label{}
  \dot{\rho}_t = \int_0^t  \sum_{\alpha=1}^{d+1} h_\alpha(t-\tau)  \mathbb{U}_\alpha[\rho_\tau] \der \tau - f(t) \rho_0 .
\end{equation}

\begin{example} In the qubit case $(d=2)$, one finds
\begin{equation}\label{}
  \dot{\rho}_t = \int_0^t  \sum_{\alpha=1}^{3} h_\alpha(t-\tau) \sigma_\alpha \rho_\tau \sigma_\alpha \der \tau - f(t) \rho_0 ,
\end{equation}
or, introducing the Bloch vector $x_\alpha(t) = {\rm Tr}[\rho_t \sigma_\alpha]$,
\begin{equation}\label{}
\begin{split}
  \dot{x}_\alpha(t) = \int_0^t  \left[2h_\alpha(t-\tau)-h(t-\tau)\right] &x_\alpha(\tau) d\tau \\&- f(t) x_\alpha(0),
\end{split}
\end{equation}
with $h(t)=\sum_{\alpha=1}^3h_\alpha(t)$.
\end{example}

\section{Discrete Wigner functions and classical semi-Markov evolution}

The information encoded into the density operator $\rho$ can be translated into the following $d+1$ probability distributions,
\begin{equation}\label{}
  \pi^{(\alpha)}_k = \Tr\left(P^{(\alpha)}_k \rho\right) .
\end{equation}
The probability vectors $\left(\pi^{(\alpha)}_1,\dots,\pi^{(\alpha)}_{d+1}\right)$ evolve according to the classical evolution equation
\begin{equation}\label{clas_evo}
\pi_k^{(\alpha)}(t)= \sum_{i=0}^{d-1} T_{ki}^{(\alpha)}(t)\pi_i^{(\alpha)}(0)
\end{equation}
with the stochastic (even doubly stochastic) map
\begin{equation}\label{}
  T_{ij}^{(\alpha)}(t) = \Tr(P^{(\alpha)}_i \Lambda_t[P^{(\alpha)}_j]).
\end{equation}
One easily finds
\begin{equation}\label{T_sol}
T^{(\alpha)}(t) = c_\alpha(t) \mathbb{I} +  \left[1- c_\alpha(t) \right]\mathcal{P} ,
\end{equation}
where $\mathcal{P}_{ij} = 1/d$ and
\begin{equation}\label{}
  c_\alpha(t) = \frac{d}{d-1}\left[p_0(t)+p_\alpha(t)-\frac 1d\right].
\end{equation}
If $\Lambda_t$ is the solution of the quantum memory kernel master equation with $K_t$ as in Theorem 1, then the stochastic map takes the following form,
\begin{equation}\label{T}
T^{(\alpha)}(t)=\left[1-\int_0^t\ell_\alpha(\tau)\der\tau\right] \mathbb{I}
+\int_0^t\ell_\alpha(\tau)\der\tau\ \mathcal{P}.
\end{equation}

Observe that, knowing the probability distributions $\pi_k^{(\alpha)}$, one can express the discrete Wigner function $W_\alpha$ in terms of $\pi_k^{(\alpha)}$. Therefore, it is also possible to find the time-evolution evolution of $W_\alpha$. Recall the definition of the discrete Wigner function \cite{Wootters2},
\begin{equation}\label{W}
W_\alpha=\frac 1d \Tr\left(\rho A_\alpha\right),
\end{equation}
where, after introducing $\alpha=(a_1,a_2)$, the operators $A_\alpha$ are given by
\begin{equation}\label{A}
\begin{split}
A_\alpha=\frac 12\big[(-1)^{a_1}\sigma_3+(-1)^{a_2}&\sigma_1+(-1)^{a_1+a_2}\sigma_2+\mathbb{I}\big]
\end{split}
\end{equation}
for $d=2$ and
\begin{equation}
(A_\alpha)_{kl}=\delta_{2a_1,k+l}\ e^{2\pi\imag a_2(k-l)/d},
\end{equation}
for a prime $d>2$.
Let us illustrate our claim in the following example.

\begin{example}
Calculate the discrete Wigner function for a qubit ($d=2$). From definition, one has
\begin{equation}
\begin{split}
&W_{00}(t)=\frac 14 \big(1+x_1(t)+x_2(t)+x_3(t)\big),\\
&W_{01}(t)=\frac 14 \big(1-x_1(t)-x_2(t)+x_3(t)\big),\\
&W_{10}(t)=\frac 14 \big(1+x_1(t)-x_2(t)-x_3(t)\big),\\
&W_{11}(t)=\frac 14 \big(1-x_1(t)+x_2(t)-x_3(t)\big),
\end{split}
\end{equation}
where $x_\alpha(t) = {\rm Tr}[\rho_t \sigma_\alpha]$ is the Bloch vector. Observe that the Bloch vector is related to the probability distributions $\pi^{(\alpha)}_k(t)$ as follows,
\begin{equation}
\pi^{(\alpha)}_{k}(t)=\frac 12\big[1- (-1)^k x_\alpha(t)\big]; \ \ k =1,2.
\end{equation}
Therefore, if $\pi^{(\alpha)}_k(t)$'s evolve according to the classical evolution equation (\ref{clas_evo}) with the bistochastic map (\ref{T}), then the corresponding discrete Wigner function satisfies the following evolution equation,
\begin{equation}
\mathbf{W}(t)=S(t)\mathbf{W},\qquad S(0)=\oper,
\end{equation}
with $\mathbf{W}=(W_{00},W_{01},W_{10},W_{11})$ and the bistochastic map $S(t)$. The map $S(t)$ has a simple structure,
\begin{equation}
S(t)=\frac 14
\begin{bmatrix*}[r]
s_0(t) & s_3(t) & s_1(t) & s_2(t) \\
s_3(t) & s_0(t) & s_2(t) & s_1(t) \\
s_1(t) & s_2(t) & s_0(t) & s_3(t) \\
s_2(t) & s_1(t) & s_3(t) & s_0(t)
\end{bmatrix*},
\end{equation}
where
\begin{align*}
&s_0(t)=4-\sum_{\beta=1}^3 \int_0^t\ell_\beta(\tau)\der\tau,\\
&s_\alpha(t)=\sum_{\beta=1}^3 \int_0^t\ell_\beta(\tau)\der\tau
-2\int_0^t\ell_\alpha(\tau)\der\tau.
\end{align*}
\end{example}

Now, suppose that $\Lambda_t$ obeys the quantum semi-Markov evolution defined by the quantum semi-Markov map $Q_t = \frac{1}{d-1}\sum_\alpha f_\alpha(t) \mathbb{U}_\alpha$. Using the representation (\ref{Dyson}),
\begin{equation}\label{}
  \Lambda_t = N_t + N_t \ast Q_t + N_t \ast Q_t \ast Q_t + \ldots,
\end{equation}
and the following property of $Q_t$,
\begin{equation}\label{}
  Q_t[P^{(\alpha)}_i ] = \sum_{j=0}^{d-1} q^{(\alpha)}_{ij}(t)  P^{(\alpha)}_j ,
\end{equation}
with
\begin{eqnarray}\label{}
  q^{(\alpha)}_{ij}(t) &=& \Tr\left(  P^{(\alpha)}_i Q_t[  P^{(\alpha)}_j]\right) \nonumber \\ &=& \delta_{ij}f_\alpha(t)+\frac{1-\delta_{ij}}{d-1} [f(t) - f_\alpha(t)],
\end{eqnarray}
one finds the corresponding representation of the stochastic map $T^{(\alpha)}(t)$,
\begin{equation}
T^{(\alpha)} = n + n \ast q^{(\alpha)} +  n \ast q^{(\alpha)} \ast  q^{(\alpha)} + \dots ,
\end{equation}
where
\begin{equation}\label{}
n_{ij}(t) = \Tr\left(P^{(\alpha)}_i N_t[  P^{(\alpha)}_j]\right) = g(t) \delta_{ij} .
\end{equation}
Interestingly, the map $n(t)$ is universal -- that is, it does not depend on `$\alpha$'.

\section{Examples}

\subsection{Markovian semigroup}

Let us consider the evolution provided by the following Gorini-Kossakowski-Sudarshan-Lindblad generator,
\begin{equation}
\mathcal{L}=\sum_{\alpha=1}^{d+1}\gamma_\alpha\mathcal{L}_\alpha,
\end{equation}
where $\gamma_\alpha\geq 0$. Due to (\ref{eigenvalue2}), one has
\begin{equation}
\lambda_\alpha(t)=e^{(\gamma_\alpha-\gamma)t}.
\end{equation}
Now, the memory kernel equation (\ref{memory_kernel_equation}), with $K_t$ satisfying the assumptions of Theorem 1, describes the dynamics of the Markovian semigroup if and only if
\begin{equation}\label{}
  \ell_\alpha(t) =  \frac{d(\gamma-\gamma_\alpha)e^{-\frac{d-1}{d}\gamma t}}
{(\gamma-d\gamma_\alpha)e^{-\frac{d-1}{d}\gamma t}+d}.
\end{equation}
Moreover, the Markovian semigroup is generated by the quantum semi-Markov map $Q_t$ (\ref{Q_t}) with
\begin{equation}\label{Markovian}
f_\alpha(t)=\frac{d-1}{d} \gamma_\alpha e^{-\frac{d-1}{d}\gamma t}.
\end{equation}

\subsection{Oscillatory behaviour}

Take the oscillating functions
\begin{equation}\label{osc}
\ell_\alpha(t) = \frac{\omega}{a_\alpha} \sin \omega t.
\end{equation}
One finds the map (\ref{GPC}) with the following probability vector,
\begin{equation}
p_0(t)=1-\frac{d-1}{d^2}(1-\cos\omega t)\sum_{\beta=1}^{d+1}\frac{1}{a_\beta},
\end{equation}
\begin{equation}
p_\alpha(t)=\frac{d-1}{d^2}(1-\cos\omega t)\left[\sum_{\beta=1}^{d+1}\frac{1}{a_\beta}-\frac{d}{a_\alpha}\right].
\end{equation}
This corresponds to the legitimate generalized Pauli channel if and only if
\begin{equation}\label{osc_cond}
\frac{d}{a_\beta}\leq\sum_{\alpha=1}^{d+1}\frac{1}{a_\alpha}\leq\frac{d^2}{2(d-1)}.
\end{equation}
Note that $\ell_\alpha(t)$'s in (\ref{osc}) give rise to the memory kernel $K_t$ with the following eigenvalues,
\begin{equation}
{\kappa}_\alpha(t)=-\frac{\omega^2}{a_\alpha}\cos\left(\sqrt{1-\frac{1}{a_\alpha}}\omega t\right).
\end{equation}
Observe that (\ref{osc_cond}) implies
\begin{equation}
a_\alpha\geq 2\left(1-\frac 1d\right),
\end{equation}
which means that $a_\alpha\geq 1$ for $d\geq 2$, and hence $K_t$ is always well-defined.

\subsection{Convex combination of Markovian semigroups}

Let us provide a simple generalization of the quantum channels considered in \cite{Nina, nasze},
\begin{equation}\label{CC}
\begin{split}
\Lambda_t=&\sum_{\alpha=1}^{d+1}x_\alpha e^{dt\mathcal{L}_\alpha}=\frac 1d \Bigg[(1+[d-1]e^{-dt})\oper\\&+(1-e^{-dt})\sum_{\alpha=1}^{d+1}x_\alpha \mathbb{U}_\alpha\Bigg],
\end{split}
\end{equation}
where $x_\alpha$'s form the probability vector. Although the Kraus representation of $\Lambda_t$ is relatively complicated, each of its eigenvalues depends only on one $x_\alpha$,
\begin{equation}
\lambda_\alpha(t) = e^{-dt}+\left(1-e^{-dt}\right)x_\alpha.
\end{equation}
For (\ref{CC}), we can find the time-local generator $\mathcal{L}_t$ (\ref{GEN}) with the following decoherence rates,
\begin{equation}\label{gammas}
\gamma_\alpha(t) =
\sum_{\beta=1}^{d+1}\frac{1-x_\beta} {1+\left(e^{dt}-1\right)x_\beta} -d\frac{1-x_\alpha} {1+\left(e^{dt}-1\right)x_\alpha}.
\end{equation}
This evolution belongs to the special class described by the memory kernels previously discussed in (\ref{special}) with
\begin{equation}
\ell(t)=de^{-dt},\qquad a_\alpha=\frac{1}{1-x_\alpha}.
\end{equation}

Consider the qubit case ($d=2$) and suppose that $x_2=x_1\equiv x$. For such a choice, it is possible to recover the semi-Markov map $Q_t$ (\ref{Q_t}) with
\begin{equation}\label{f1}
\begin{split}
f_1(t)&=f_2(t)=\\&\frac{x}{\xi}e^{-\frac{3-2x}{2}t}\Bigg[
\xi\cosh\frac{\xi t}{2}-3(1-2x)\sinh\frac{\xi t}{2}\Bigg],
\end{split}
\end{equation}
\begin{equation}\label{f3}
\begin{split}
f_3(t)=\frac{1}{\xi}e^{-\frac{3-2x}{2}t}\Bigg[
\xi(1-2x)&\cosh\frac{\xi t}{2}\\&-(4x-1)\sinh\frac{\xi t}{2}\Bigg],
\end{split}
\end{equation}
where $\xi:=\sqrt{12x^2-4x+1}$. Note that
\begin{equation}
\begin{split}
\int_0^\infty\big[f_1(t)+f_2(t)&+f_3(t)\big]\der t\\&=\frac{-3x^2+3x-1}{x^2+x-1} \leq 1,
\end{split}
\end{equation}
and hence the evolution is semi-Markov if and only if $f_\alpha(t)\geq 0$ for all $t\geq 0$. This holds for the Markovian semigroup ($x=0$) and for the maximally mixed probability vector, i.e. $x=1/3$.

It turns out that this property carries over to higher dimensions. Indeed, for the probability vector $x_\alpha=\frac{1}{d+1}$, the generalized Pauli channel $\Lambda_t$ in (\ref{CC}) is generated by the quantum semi-Markov map $Q_t$ with
\begin{equation}\label{CC_Markovian}
f_\alpha(t)=\frac{d-1}{d+1}e^{-\frac{d(d+1)-1}{d+1}t}.
\end{equation}
Note that this evolution is also Markovian, as (\ref{gammas}) simplifies to
\begin{equation}
\gamma_\alpha(t)=\frac{d}{d+e^{dt}}.
\end{equation}
Observe that, in this case, the semi-Markov evolution is a subclass of the Markovian evolution.

\subsection{Eternally non-Markovian evolution}

As another special case of the convex combination of the Markovian semigroups, we analyze the eternally non-Markovian evolution, where $x_\alpha=1/d$ (for $\alpha=1,\dots,d$) and $x_{d+1}=0$. This corresponds to the following choice of the decoherence rates,
\begin{equation}\label{EM_gammas}
\gamma_\alpha(t)=1,\qquad
\gamma_{d+1}(t)=-(d-1)\frac{e^{dt}-1}{e^{dt}-1+d},
\end{equation}
with $\gamma_{d+1}(t)\leq 0$ for all $t\geq 0$. For $d=2$, one recovers the well-known eternally non-Markovian evolution of the qubit \cite{ENM},
\begin{equation}
\gamma_1(t)=\gamma_2(t)=1,\qquad \gamma_3(t)=-\tanh t.
\end{equation}
To determine whether this evolution is semi-Markov, we find the map $Q_t$ (\ref{Q_t}) with
\begin{equation}
\begin{split}
&f_\alpha(t)=(d-1)e^{-dt/2}\Bigg[
\frac{1}{d}\cosh\left(\sqrt{\frac{d^3-4d+4}{4d}}t\right)\\&
-\frac{d-2}{\sqrt{d^4-4d^2+4d}}\sinh\left(\sqrt{\frac{d^3-4d+4}{4d}}t\right)\Bigg],\\
\end{split}
\end{equation}
for $\alpha=1,\ldots,d$, and 
\begin{equation}
\begin{split}
f_{d+1}(t)&=-\frac{2(d-1)^2}{\sqrt{d^4-4d^2+4d}}e^{-dt/2}\\&\times
\sinh\left(\sqrt{\frac{d^3-4d+4}{4d}}t\right).
\end{split}
\end{equation}
Note that $f_\alpha(t)\geq 0$ for  $\alpha=1,\ldots,d$ but $f_{d+1}(t) < 0$ (for $t>0$). Therefore, $Q_t$ is not completely positive, and hence  the corresponding dynamical map $\Lambda_t$ is not semi-Markov. This example shows that, in general, the convex combination of Markovian semigroups goes beyond the semi-Markov evolution.

\section{Conclusions}

Using the memory kernel master equation, we analyzed the evolution of the special class of the dynamical maps, provided by the generalized Pauli channels.
We found the necessary and sufficient conditions which guarantee that the corresponding solution defines the legitimate physical evolution (CPTP map). Moreover, we analyzed a special class of the kernels corresponding to the quantum semi-Markov evolution. Such evolution defines a generalization of the Markovian semigroup. Surprisingly, the convex combination of Markovian semigroups is not semi-Markov. Several examples illustrate the general approach.

It would be interesting to further analyze the memory kernels going beyond the semi-Markov case. The example of the eternally non-Markovian evolution shows that one can obtain the legitimate dynamical map
\begin{equation}
\Lambda_t = N_t + N_t \ast Q_t + N_t \ast Q_t \ast Q_t + \ldots
\end{equation}
from not completely positive $Q_t$. Therefore, one would like to find weaker conditions for the maps $N_t$, $Q_t$ that still guarantee the complete positivity of $\Lambda_t$.

\section*{Acknowledgements} This paper was partially supported by the National Science Centre project 2015/17/B/ST2/02026.

\end{document}